\documentclass[]{article}

\usepackage{graphicx}
\usepackage{epstopdf}
\usepackage{epsf}
\usepackage{epsfig}
\usepackage{amsmath}
\usepackage{amssymb}
\usepackage{dsfont}
\usepackage{comment}
\usepackage{eucal}

\newcommand{\invt}[2]{c_{#1}(#2)}

\newtheorem{theorem}{Theorem}[section]
\newtheorem{lemma}[theorem]{Lemma}

\newenvironment{proof}[1][Proof]{\begin{trivlist}
\item[\hskip \labelsep {\bfseries #1}]}{\end{trivlist}}

\newcommand{\qed}{\nobreak \ifvmode \relax \else
      \ifdim\lastskip<1.5em \hskip-\lastskip
      \hskip1.5em plus0em minus0.5em \fi \nobreak
      \vrule height0.75em width0.5em depth0.25em\fi}

\newcommand{\eqnref}[1]{(\ref{#1})}

\title{Nerio: % \footnote{Nerio (strength) is the wife of Mars}~:
Leader Election and Edict Ordering}
\author{Robbert van Renesse, Fred B. Schneider, Johannes Gehrke}
\date{Department of Computer Science, Cornell University}
\begin{document}

\maketitle

\section{Introduction}

Coordination in a distributed system is facilitated if there is a
unique process,
the \emph{leader}, to manage the other processes.
The leader creates \emph{edicts} and sends them to other processes
for execution or forwarding to other processes.
The leader may fail, and when this occurs
a leader election protocol selects a replacement.
That protocol satisfies the following properties:

\begin{itemize}
\item \emph{Leader Uniqueness}: At any time, at most one process
is \emph{leader}.
\item \emph{Edict Validity}: Only leaders can create edicts.
\item \emph{Edict Ordering}: Recipients of multiple edicts
can determine the real-time order in which the edicts were created.
\item \emph{Leader Stability}: If a process is leader then, in the absence
of failures and in the presence of timely communication and processing,
it remains leader.
\item \emph{Eventual Election}: If there is no leader then, in the presence
of sufficient timely communication and processing and a bounded number of
failures, a leader is elected.
\item \emph{Fault Tolerance}: Some number of crash failures are tolerated.
\item \emph{Efficiency}: The time and storage, processing, and networking
resources required by the protocol are reasonable.
\end{itemize}

We assume that processes can exhibit failures: a process
operates correctly until a failure causes that process to stop taking
execution steps.
Crashed processes are assumed to maintain the values of their variables
although these variables are no longer accessible.
The clock at a process is a variable and the exposition is simplified if
that clock continues to advance even after the process has failed (and
the clock is no longer accessible).
Message delivery latencies and processing times are assumed to be unbounded.
Message loss and reordering by the network is allowed, and network
partitioning is permitted too.

\emph{Nerio} is a class of leader election protocols that implement
these properties,
Besides developing this class, we derive refinements for two plausible
environments: one
assumes bounded drift of clock rate with respect to the rate of real time;
the second assumes bounded differences between clock values
on any two processes at the same time.

Nerio protocols are based on granting leases~\cite{GC89} and require
that failure scenarios are characterized by
quorum systems~\cite{Tho78}, a combination first found in
the leader election protocol of Fetzer and S\"{u}{\ss}kraut in~\cite{FS06}.
But leader election properties (Leader Uniqueness, Leader Stability,
and Eventual Election) alone offer little value,
since a leader may no longer be the leader by the time
it sends a message
let alone when such a message is received by another process.
Our Nerio protocols, which in addition satisfy
Edict Validity and Edict Ordering properties, do provide value in
asynchronous environments because
edicts sent by leaders can be interpreted in the order of their creation,
even if the processes that sent the edicts have ceased being leaders.

\subsection*{Formalizing the Properties}

Consider a finite set of processes $P = \{p, ...\}$.
Let $\textit{isLeader}_p(t)$ be the property that, at time $t$,
process $p$ is leader.
Formally, Leader Uniqueness is the following:

\begin{itemize}
\item[] \textit{Leader Uniqueness}:
\begin{equation}\label{eq:uniqueness}
\forall p, q \in P, t: (\textit{isLeader}_p(t) \wedge \textit{isLeader}_q(t))
\Rightarrow (p = q).
\end{equation}
\end{itemize}

A leader can create \emph{edicts} that it sends to other
processes.
For an edict $e$, define $e.\textit{creator}$ to be the process that
created $e$, and
$e.\textit{created}$ to be the real time at which $e$ is created.
(Note that even the process itself cannot know this time.)
Only leaders can create edicts:

\begin{itemize}
\item[] \textit{Edict Validity}:
\begin{equation}\label{eq:edictvalidity}
\forall e: \textit{isLeader}_{e.\textit{creator}}(e.\textit{created})
\end{equation}

\end{itemize}
Edict Ordering means that
\begin{itemize}
\item There is a total ordering $\prec$ on edicts.
\item For edicts $e$ and $e'$, $e.\textit{created} < e'.\textit{created}
						\Rightarrow e \prec e'$.
\item Any recipient of edicts $e$ and $e'$ can ascertain whether $e \prec e'$
or $e' \prec e$ holds.
\end{itemize}

However, Edict Ordering does not imply that receivers all receive the same
set of edicts.
Let $\textit{Order}_p(e_1, e_2)$ mean that process $p$ received edicts
$e_1$ and $e_2$, and believes that $e_1$ was created before $e_2$.
Formally, Edict Ordering is the following:

\begin{itemize}
\item[] \textit{Edict Ordering}:
\begin{equation}\label{eq:edictordering}
\forall p, e_1, e_2: \textit{Order}_p(e_1, e_2) \Leftrightarrow
				e_1.\textit{created} < e_2.\textit{created}
\end{equation}
\end{itemize}

In order to formalize Leader Stability and Eventual Election formally, we
assume that there is a time after which message latencies between
correct processes are bounded by a known constant $d$, and there are no
more failures.
We call this the Global Stabilization Time (\texttt{GST}).
We do not know when \texttt{GST} is, only that it will happen eventually.
Then we can have the following properties:

\begin{itemize}
\item[] \textit{Leader Stability}:
\begin{equation}
\exists \textit{GST}: \forall t_1, t_2 > \texttt{GST}, ~ p \in P: (\textit{isLeader}_p(t_1) \wedge
t_1 < t_2) \Rightarrow \textit{isLeader}_p(t_2)
\end{equation}
\item[] \textit{Eventual Election}:
\begin{equation}
\exists \textit{GST}: \exists t > \texttt{GST}, p \in P: \textit{isLeader}_p(t)
\end{equation}
\end{itemize}

In Section~\ref{sec:class}, we describe the Nerio class of leader election
protocols that leverage the properties of quorum systems instead of
requiring accurate failure detection.
Section~\ref{sec:drift} describes a protocol in this class; it
assumes bounded clock drift.
Section~\ref{sec:skew} describes another protocol that assumes
that there is a bound on how much two clocks may differ.
We compare the two protocols in Section~\ref{sec:compare}.
In Section~\ref{sec:release} we show how a process can give up its grants
to a lease if so desired.
Section~\ref{sec:edicts} shows how Nerio protocols support
Edict Validity and Edict Ordering.
We show that the protocols satisfy Leader Stability in
Section~\ref{sec:stability}, while Section~\ref{sec:eventual}
demonstrates that the protocols satisfy Eventual Election.
A discussion of various issues follows in Section~\ref{sec:discussion}.
Section~\ref{sec:related} discussion prior work.

\section{A Class of Leader Election Protocols}\label{sec:class}

Let $\cal Q$ be a quorum system on $P$.  That is: $\cal Q$ is a set of
process sets such that

\begin{equation}
\forall Q \in {\cal Q}: Q \subseteq P
\end{equation}
\begin{equation}\label{eq:overlap}
\forall Q_1, Q_2 \in {\cal Q}: Q_1 \cap Q_2 \ne \emptyset
\end{equation}

\noindent
An oft-used quorum system consists of all subsets that are majorities in $P$,
that is, $\forall Q \in {\cal Q} \Rightarrow |Q| > |P| / 2$.

Each process $p$ has the following state variables (we use upper case
characters to denote local variables):

\begin{itemize}
\item[] $C_p$ (clock): a monotonically increasing clock at process $p$;
\item[] $A_p$ (assignee): a process, initially $p$ itself;
\item[] $F_p$ (finish): a clock value measured on the clock of process $p$, initially~0;
\item[] $E_p$ (expiration): another clock value measured on the clock of process $p$, initially~0.
\end{itemize}

\noindent
If $X_p$ is a local variable at process $p$, then we write
$X_p(t)$ for the value of $X_p$ at real time $t$.

Assume that $C_p(t)$ is continuous and
satisfies the following two conditions, which should hold for
the clocks found on real processes:

\begin{itemize}
\item[] \textit{Monotonicity}:
\begin{equation}
\forall t_1, t_2: t_1 < t_2 \Rightarrow C_p(t_1) < C_p(t_2)
\end{equation}
\item[] \textit{Growth}:
\begin{equation}\label{eq:growth}
\forall T > 0: \exists t: C_p(t) \ge T
\end{equation}
\end{itemize}

In practice, the hardware clock increases in a stepwise fashion
rather than continuously.
This is not observable if a clock has a sufficiently high resolution relative
to the speed at which processes advance.
A process can only sample its clock, so by obtaining a value $T$
the process only learns that between the time that the process
requested the sample and the time that it obtained the sample, the
value of the clock was $T$.
We assume that the clock advances from one sample to the next.  This
can be ensured by making the clock a pair consisting of the hardware
clock and a counter that is reset each time the hardware clock advances
and is incremented each time the clock is sampled.  This composite
clock is then ordered lexicographically.
Because of the asynchronous nature of our system,
an arbitrary interval may have elapsed between when the sample
sample is taken and when it is returned to the process.
Therefore, a process cannot tell the difference between a clock
that increases continuously, and one that does not.

Assuming that $C_p$ increases we can
define an inverse function $\invt{p}{T}$ on clocks with the following
properties:

\begin{equation}
C_p(\invt{p}{T}) = T
\end{equation}
\begin{equation}
\invt{p}{C_p(t)} = t
\end{equation}

\begin{lemma}\label{lemma:clock}
\[
\forall p \in P, t, T: C_p(t) < T \Leftrightarrow t < \invt{p}{T}
\]
\end{lemma}

% Each process $p$ is allowed to increase $F_p$ at any time and by any
% (positive) amount, but it is never allowed to decrease it.
% That is $t_1 > t_2 \Rightarrow F_p(t_1) \ge F_p(t_2)$.
% A process is only allowed to change $A_p$ if $C_p \ge F_p$ (and
% hence not before $\invt{p}{F_p}$ in real time).

\noindent Let $\gamma_{p,q}(t)$ be the predicate 
\[
\gamma_{p,q}(t) \;\; \equiv \;\; A_q(t) = p ~\wedge~ C_q(t) < F_q(t).
\]
\noindent If $\gamma_{p,q}(t)$ holds, we say that,
at time $t$, \emph{process $q$ grants a lease to process $p$}.  Note
that a process cannot grant a lease to two different processes at the
same time $t$, because a variable (\emph{e.g.}, $A_q(t)$) can have only
one value at time $t$.
% Also, if $q$ grants a lease at time $t$, then that grant will
% hold at least until $C_q$ reaches $F_q$, as $q$ is only allowed to
% increase $F_q$ and is not allowed to change $A_q$ until $C_q \ge F_q$
% (\emph{i.e.}, not before $\invt{q}{F_q}$).

We can now define formally what it means to be leader:

\begin{equation}\label{eq:quorum}
\textit{isLeader}_p(t) \equiv
\exists Q \in {\cal Q}: (\forall q \in Q: \gamma_{p,q}(t))
\end{equation}

\noindent
That is, process $p$ is leader at time $t$ iff
a quorum of processes grant a lease to $p$ at time $t$.
It should be clear to the reader that \eqnref{eq:quorum} implies
\eqnref{eq:uniqueness}: because of the intersection property of
quorums (Equation~\eqnref{eq:overlap})
there cannot be two different quorums, one in which all processes are
granting a lease to $p_1$,
and another quorum in which all processes are granting a lease to
a different process $p_2$, at the same time.

We need an implementation of $\textit{isLeader}_p$.
Each process $p$ has a variable $E_p$, which gives an
expiration time of $p$'s leadership, initially 0.
Like $F_p$, $E_p$ is measured on $p$'s clock.
In Nerio protocols, the following invariant holds:

\begin{equation}\label{eq:isleader}
\forall p, t: (C_p(t) < E_p(t)) \Rightarrow \textit{isLeader}_p(t) 
\end{equation}

(The implication holds only in one direction because, as we shall see,
processes extend their grants conservatively, and thus it may be that
a quorum of processes are still granting a lease to $p$ after $p$ gives up
on the lease.)

Combining~\eqnref{eq:isleader} and~\eqnref{eq:quorum}
and substituting $\gamma_{p,q}(t)$, we get the following property:

\begin{equation}\label{eq:invariant}
\forall p, t: (C_p(t) < E_p(t)) \Rightarrow
(\exists Q \in {\cal Q}: (\forall q \in Q: A_q(t) = p \wedge C_q(t) < F_q(t)))
\end{equation}

We take this as the defining characteristic of a Nerio class leader
election protocol.

At this point it is useful to consider what happens if a process
crashes.
By the \textit{Growth} condition (Equation~\eqnref{eq:growth}),
the clock of the process continues increasing.
We need this in order to ensure that if a crashed process $p$ was a leader,
eventually it stops being leader (because $C_p(t) < E_p(t)$ becomes
false), and if a crashed process $q$ granted a lease, eventually
this lease expires (because $C_q(t) < F_q(t)$ becomes false).
Since a crashed process cannot produce any output, having the clock
stop is indistinguishable from a clock that continues to increase.\footnote{
In practice, a hardware clock often continues to increase for some
amount of time as it is backed up by an internal battery.}

Below we will show examples of protocols that maintain~\eqnref{eq:invariant},
given certain assumptions about the environment.

\section{Clocks with Bounded Drift}\label{sec:drift}

Assume the drift (accuracy of rate) of each clock is bounded by a
constant $\rho$ per time unit.  That is:

\begin{equation}\label{eq:drift}
\forall p \in P, t, \delta: C_p(t) + (1 - \rho) \delta \le C_p(t + \delta) \le C_p(t) + (1 + \rho) \delta.
\end{equation}

\noindent
In other words, during a real-time period $\delta$, the clock of a process
may advance by as little as $(1 - \rho) \delta$, or as much as
$(1 + \rho) \delta$.

In the Nerio class protocol that we derive in this section,
a process $p$ never decreases $F_p$.
Thus the protocol maintains the following invariant:

\begin{equation}\label{eq:monotonic}
\forall p \in P, t_1, t_2: t_1 < t_2 \Rightarrow F_p(t_1) \le F_p(t_2).
\end{equation}

\noindent
Furthermore, consistent with the meaning of a lease,
a process $p$ never changes $A_p$ if $C_p < F_p$.
As a result, once a process $p$ has granted a lease to $A_p$, this
grant remains until real time $\invt{p}{F_p}$.

A process $p$ trying to become leader (or extend the period
during which it is leader) executes the following algorithm, which
we call \texttt{Obtain Quorum with Bounded Drift}, or \texttt{OQwBD}
for short.
Process $p$ uses a temporary $\textit{Start}_p$ into which it stores
the starting time of the algorithm:

\begin{enumerate}
\item set $\textit{Start}_p := C_p$ (sample starting time);
\item select a real time period $\delta$, $\delta > 0$;
\item broadcast $\langle \texttt{grantRequest}, p, \textit{Start}_p, \delta \rangle$.
\end{enumerate}

\noindent Upon receipt of a \texttt{grantRequest} message,
a process $q$ does the following:

\begin{enumerate}
\setcounter{enumi}{3}
\item $T_q := C_q$ (save local time into a temporary variable $T_q$);
\item if $p \ne A_q \wedge T_q < F_q$, then ignore the request
($q$ is already granting a lease to $A_q$, $A_q \ne p$);
\item otherwise
\begin{itemize}
\item[6.1.] $A_q := p$; $F_q := \max(F_q, T_q + (1 + \rho) \cdot \delta)$;
\item[6.2.] send $\langle \texttt{ok}, q, \textit{Start}_p \rangle$ to $p$.
\end{itemize}
\end{enumerate}

\noindent
Meanwhile, process $p$ waits for \texttt{ok} messages:

\begin{enumerate}
\setcounter{enumi}{6}
\item wait for a
$\langle \texttt{ok}, q, \textit{Start}_p \rangle$ from each process $q$
in a quorum of $\cal Q$ or until
$C_p \ge \textit{Start}_p + (1 - \rho) \cdot \delta$;
\item if \texttt{ok} messages are received from a quorum and $C_p < \textit{Start}_p + (1 - \rho) \cdot \delta$,
then $E_p := \textit{Start}_p + (1 - \rho) \cdot \delta$
(we say that the \texttt{OQwBD} algorithm \emph{completed});
\item if not a sufficient number of \texttt{ok} responses are received
by $C_p \ge \textit{Start}_p + (1 - \rho) \cdot \delta$,
then this instantiation of \texttt{OQwBD} \emph{failed}.
\end{enumerate}

\noindent
By measuring $(1 - \rho) \cdot \delta$ on its local clock,
process $p$ will stop believing it is leader before \emph{at most}
$\delta$ real time units
have expired since $p$ initiated \texttt{OQwBD}.
A process $q$, by measuring $(1 + \rho) \cdot \delta$, grants the lease
for \emph{at least} $\delta$ real time units since process $p$ started
\texttt{OQwBD}.
(Process $q$ calculates a maximum in order to ensure that $F_q$ can only
progress forwards as required by~\eqnref{eq:monotonic}).

Process $p$ can run \texttt{OQwBD} at any time.
We say that $p$ \emph{aborts} \texttt{OQwBD}
if $p$ starts a new execution before the current one completed.
Once aborted, responses for the earlier instantiation of
\texttt{OQwBD} will be ignored.
The clock value $\textit{Start}_p$ is included in the messages only to
identify an instantiation of \texttt{OQwBD}
(the tuple $(p, \textit{Start}_p)$
uniquely identifies an instantiation of \texttt{OQwBD});
receivers do not interpret the clock value, but return it in the response.
This way, process $p$ can ignore responses of aborted instantiations.

\begin{figure}
\begin{center}
\includegraphics[height=5in]{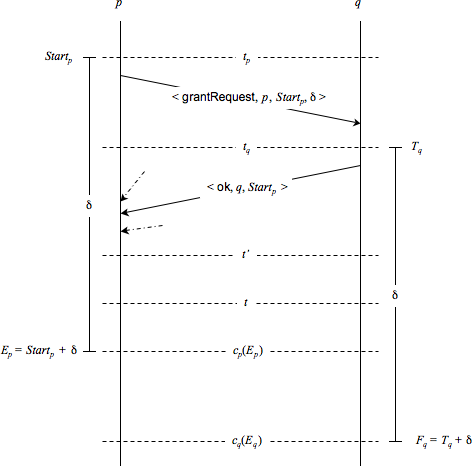}
\end{center}
\caption{\label{fig:leader} Example of a protocol exchange between a
process $p$ that initiated the protocol and a process $q$ that granted
the lease and is included in the quorum that
$p$ uses to complete the protocol.  For clarity, the drift $\rho = 0$.
Dashed horizontal lines indicate
real time.  The labels to the left are clock values on $p$'s clock;
the labels on the right are clock values on $q$'s clock.}
\end{figure}

To prove that~\eqnref{eq:invariant} holds for any $p$ and $t$,
consider a process $p$ and a time $t$ for which $C_p(t) < E_p(t)$ holds,
and note that $t < \invt{p}{E_p(t)}$ (Lemma~\ref{lemma:clock}).
Let $t'$ be the time at which the last instantiation of \texttt{OQwBD}
completed at process $p$.
Note that $p$ updates $E_p$ only at this time, and since this is
the last instantiation of the protocol, we have the following:

\begin{equation}\label{eq:finish}
E_p(t) = E_p(t')
\end{equation}

Let $t_p$ be the time at which $p$ assigned $\textit{Start}_p$ in this
last instantiation of \texttt{OQwBD}, and thus

\begin{equation}\label{eq:start}
\textit{Start}_p(t') = \textit{Start}_p(t_p)
\end{equation}

We have to show that there exists a quorum $Q$ so that
$\forall q \in Q: A_q(t) = p ~\wedge~ C_q(t) < F_q(t)$.
We show that the quorum that responded to $p$ and caused $p$
to complete \texttt{OQwBD} is such a quorum.

Let $Q$ be the quorum that responded to $p$.
Consider a process $q \in Q$ and
let $t_q$ be the time at which some process $q$ sampled the local clock
resulting in a value $T_q$, so that $T_q = C_q(t_q)$.  Note that
\begin{equation}\label{eq:ground}
t_p \le t_q \le t' \le t < \invt{p}{E_p(t)}.
\end{equation}
(See Figure~\ref{fig:leader} for an illustration in the case $\rho = 0$.)

\begin{lemma}\label{lemma:oqwbd}
$\invt{p}{E_p(t)} \le \invt{q}{F_q(t)}$.
\end{lemma}
\begin{proof}
\begin{align*}
(1) \;\; & E_p(t) = E_p(t') & (\textrm{Equation~\eqnref{eq:finish}}) \\
(2) \;\; & E_p(t') = \textit{Start}_p(t') + (1 - \rho) \delta & (\textrm{Algorithm \texttt{OQwBD}}) \\
(3) \;\; & \textit{Start}_p(t') = \textit{Start}_p(t_p) & (\textrm{Equation~\eqnref{eq:start}}) \\
(4) \;\; & \textit{Start}_p(t_p) = C_p(t_p) & (\textrm{Algorithm \texttt{OQwBD}}) \\
(5) \;\; & C_p(t_p) + (1 - \rho) \delta \le C_p(t_p + \delta) &  (\textrm{Equation~\eqnref{eq:drift}}) \\
(6) \;\; & E_p(t) \le C_p(t_p + \delta) & (\textrm{Combining (2) thru (5)}) \\
(7) \;\; & \invt{p}{E_p(t)} \le t_p + \delta & (\textrm{Lemma~\ref{lemma:clock}}) \\
(8) \;\; & F_q(t) \ge F_q(t_q) & (\textrm{Equations~\eqnref{eq:ground} and~\eqnref{eq:monotonic}}) \\
(9) \;\; & F_q(t_q) = C_q(t_q) + (1 + \rho)\delta & (\textrm{Algorithm \texttt{OQwBD}}) \\
(10) \;\; & C_q(t_q) + (1 + \rho)\delta \ge C_q(t_q + \delta) & (\textrm{Equation~\eqnref{eq:drift}}) \\
(11) \;\; & F_q(t) \ge C_q(t_q + \delta) & (\textrm{Combining (8), (9), (10)}) \\
(12) \;\; & t_q + \delta \le \invt{q}{F_q(t)} & (\textrm{Lemma~\ref{lemma:clock}}) \\
(13) \;\; & (t_p + \delta) \le (t_q + \delta) & (\textrm{Equation~\eqnref{eq:ground}}) \\
(14) \;\; & \invt{p}{E_p(t)} \le (t_p + \delta) \le (t_q + \delta) \le \invt{q}{F_q(t)} & (\textrm{Combining (7), (12), (13)})
\end{align*}
\hfill\qed
\end{proof}

It remains to show that $A_q(t) = p \wedge C_q(t) < F_q(t)$.
This follows directly from
$t_q \le t' \le t < \invt{p}{E_p(t)} \le \invt{q}{F_q(t)}$.
After assigning $A_q$ and $F_q$ between $t_q$ and $t'$, $A_q$ cannot be
changed until $\invt{q}{F_q(t)}$ at the earliest.

Note that no effort is made to detect process crashes.
If $\textit{isLeader}_p(t)$ at the time a process $p$ crashes, that process
continues to be leader until there is no longer a quorum of processes
that grant a lease to $p$.

\section{Clocks with Bounded Skew}\label{sec:skew}

Our second instance of a Nerio class leader election protocol is similar
to the first, but instead of
assuming bounded drift (Equation~\eqnref{eq:drift}) we assume that
clocks at any two processes always differ by at most $\Delta$:

\begin{equation}\label{eq:skew}
\forall p, q, t: -\Delta \le C_p(t) - C_q(t) \le \Delta.
\end{equation}

We present a new algorithm called \texttt{Obtain Quorum with Bounded
Skew}, or \texttt{OQwBS} for short.
(Skew is the difference between two clock values at the same time.)
The variables of \texttt{OQwBS} are the same as those of \texttt{OQwBD}.
A process $p$ can initiate \texttt{OQwBS} as follows:

\begin{enumerate}
\item set $\textit{Start}_p := C_p$ (sample starting time);
\item select a time period $\delta$, $\delta > 0$;
\item broadcast $\langle \texttt{grantRequest}, p, \textit{Start}_p, \delta \rangle$.
\end{enumerate}

\noindent Upon receipt, a process $q$ does the following:

\begin{enumerate}
\setcounter{enumi}{3}
\item $T_q := C_q$ (sample local time);
\item if $p \ne A_q \wedge T_q < F_q$, then ignore the request;
\item otherwise
\begin{itemize}
\item[6.1.] $A_q := p$; $F_q := \max(F_q, \textit{Start}_p + \delta + \Delta)$;
\item[6.2.] send $\langle \texttt{ok}, q, \textit{Start}_p \rangle$ to $p$.
\end{itemize}
\end{enumerate}

Note that because any two clocks differ by at most $\Delta$, 
$q$ can interpret $\textit{Start}_p$ with respect to its own clock.
As before, process $p$ waits for \texttt{ok} responses:

\begin{enumerate}
\setcounter{enumi}{6}
\item wait for a
$\langle \texttt{ok}, q, \textit{Start}_p \rangle$ from each process $q$
in a quorum of $\cal Q$ or until
$C_p \ge \textit{Start}_p + \delta$;
\item if \texttt{ok} messages are received from a quorum and $C_p < \textit{Start}_p + \delta$,
then $E_p := \textit{Start}_p + \delta$
(we say that the \texttt{OQwBS} algorithm \emph{completed});
\item if not a sufficient number of \texttt{ok} responses are received
by $C_p \ge \textit{Start}_p + \delta$,
then this instantiation of \texttt{OQwBS} \emph{failed}.
\end{enumerate}

Again, the proof of Leader Uniqueness is based on showing that
$\invt{p}{E_p(t)} \le \invt{q}{F_q(t)}$, and it is easy to see why this is true.

\section{Comparison}\label{sec:compare}

In the \texttt{OQwBD} protocol of Section~\ref{sec:drift},
a process $q$ may grant a lease for a process $p$
long beyond $\invt{p}{E_p(t)}$ if the \texttt{grantRequest} message
to $p$ is delayed that much.
If $p$ has failed, this grant could prevent other processes from becoming
leader.  It thus appears that the \texttt{OQwBS} protocol of
Section~\ref{sec:skew}, which is based on bounded skew,
has an important advantage.
However, below we will argue that in practice bounded drift is
more likely to be guaranteed than bounded skew,
so \texttt{OQwBD} is likely to be more robust in practice.

Hardware clock manufacturers often specify a bound on clock drift,
and this bound is typically within the range of $10^{-7}$ to $10^{-5}$
given a sufficiently stable temperature within the
casing of a computer chassis.
For performance measurements, in which it is necessary
to measure the passage of time, rather than to tell what time it is,
operating systems usually provide access to the raw clock value,
as opposed to one that may be adjusted by a clock synchronization protocol
attempting to reduce skew.

Under virtualization, the hardware clock may be virtualized, and drift
would no longer be bounded.
Fortunately, Xen allows guests to sample the hardware clock.  Under VMware,
the hardware clock is not directly accessible.  Fortunately, VMware does
make CPU performance counters accessible, including a way to measure the
passage of time.
If this facility documents a bound on drift, then
this is enough for our purposes.
However, if a virtual machine is migrated, a clock may jump
arbitrarily, violating the assumptions that we make on clocks.

The protocol based on bounded skew allows processes to
leverage a bound $\Delta$ to avoid a process granting a lease more
than $\Delta$ beyond $\invt{p}{E_p(t)}$.
Bounded skew requires a \emph{clock synchronization algorithm}.
Clock synchronization algorithms require bounded
latency on communication and bounded execution times,
in addition to requiring bounded clock drift.
In the absence of such bounds, clock synchronization algorithms such as
NTP provide, at best, probabilistic bounds on skew
(with unspecified probability).

Below we will only assume bounded clock drift and not bounded skew,
although the results are generalized easily.

\section{Releasing Grants}\label{sec:release}

It is sometimes useful for a process to give up the grants it received.
For example, if a process is not able to obtain grants from a quorum,
and thus does not have a lease on leadership, then it might as well
give up the grants that it has so that perhaps another process can be
more lucky.  Even if a process did obtain a lease and became leader,
it may for some reason give up its leadership by releasing its grants.
In this section we will show how this can be done without violating
invariant~\ref{eq:invariant}.

A process $p$ that wants to release its grants first aborts any instance of
\texttt{OQwBD} that it may be running.
Second, process $p$ sets $E_p(t)$ to $C_p(t)$.
We note that it is always safe for a process~$p$ to do so
as this cannot affect the validity of invariant~\ref{eq:invariant}.
If $p$ was leader, it will no longer be leader as a result.
So at this point, $p$ is neither leader nor is it trying to become one.

Next $p$ broadcast a request to all peers to release its grant.
A process $q$ that receives such a message from $p$ will check to see if
$p = A_q$.  If not, it ignores the request.
If so, it will set $F_q$ to $C_q$, causing the grant to expire immediately,
\emph{even if} $F_q > C_q$.
The reader will notice that this violates invariant~\ref{eq:monotonic}, which
was used in Lemma~\ref{lemma:oqwbd} to proof that
$\invt{p}{E_p(t)} \le \invt{q}{F_q(t)}$.
However, this lemma was only shown to hold when $C_p(t) < E_p(t)$, and
because $p$ has reset $E_p$ to $C_p$, this precondition no longer holds.
Invariant~\ref{eq:monotonic} is not used elsewhere, and thus expiring the
grant does not violate invariant~\ref{eq:invariant}.

\section{Edicts}\label{sec:edicts}

Leaders create edicts, which they send to one or more processes.
Because of a lack of assumptions about message latencies and process
execution speeds, such
an edict may take an arbitrary amount of time to arrive at a process,
and may even be lost in the network.
Edicts may also be stored or forwarded.
So an old edict may be delivered after an edict
that was created more recently, possibly by a different leader.
Edict Ordering prevents chaos: it ensures that any
two different edicts can be compared and ordered in a manner consistent
with the real times of their creation.

For this ordering to make sense, the time an edict is created must be
defined properly.
The last step by a process $p$ creating an edict $e$ is to sample
$C_p$, obtaining a value $T = C_p(t)$.
In order to ensure Edict Validity, the process determines if
$T < E_p$.
If so, then $t$ is the creation time of edict $e$, that is,
$e.\textit{created} = t$.
(Unfortunately, even the leader itself cannot determine $t$.)
If not, then the edict creation fails, because at time $t$,
process $p$ may not have been leader.

We describe how $\textit{Order}_p(e_1, e_2)$ in
Equation~\eqnref{eq:edictordering}
is implemented for algorithm \texttt{OQwBD},
but the idea does not depend on specifics of \texttt{OQwBD}.
Extend the \texttt{ok} response (Step 2) from $q$ with $T_q$, such that $q$
sends $\langle \texttt{ok}, q, T_q, \textit{Start}_p \rangle$ to $p$.
Process $p$ awaits messages from all processes in a quorum $Q \in {\cal Q}$,
and constructs a \emph{Quorum Timestamp} $\textit{QT}_p$ as the
set of pairs $(q, T_q)$ for all $q \in Q$.
In addition, process $p$ maintains an \emph{Edict Counter} $\textit{EC}_p$,
initially~0, counting the number of edicts created by $p$.

Every time $p$ creates an edict, it tags that edict with
an \emph{Edict Timestamp} $(\textit{QT}_p, \textit{EC}_p)$ and
increments $\textit{EC}_p$.
Before sending the edict, process $p$ checks to see if $C_p < E_p$ to
make sure it is still leader.
If not, the edict is not valid and should be discarded.

We define an ordering on Edict Timestamps and show it
consistent with the real time in which the edicts were created.
Edict timestamps are lexicographically ordered, first by quorum timestamp
and then by the natural ordering on edict counters.
Quorum timestamps are ordered as follows:

\begin{equation}
\textit{QT}_1 < \textit{QT}_2 \Leftrightarrow
(\exists q, T_1, T_2: (q, T_1) \in \textit{QT}_1 \wedge (q, T_2) \in \textit{QT}_2
\wedge T_1 < T_2)
\end{equation}

We show that this ordering is consistent with the creation times of edicts.
Let $X$ be a completed instantiation of \texttt{OQwBD}.  $X$ has the
following attributes:

\vspace{1em}
\begin{tabular}{l|l}
$X.\textit{owner}$ & the process that initiated $X$ and became leader \\ \hline
$X.\textit{start}$ & the real-time when $X$ started (\emph{i.e.}, $\invt{X.\textit{owner}}{\textit{Start}_{X.\textit{owner}}}$)\\ \hline
$X.\textit{completion}$ & the real-time when $X$ completed \\ \hline
$X.\textit{QT}$ & the quorum timestamp that $X.\textit{owner}$ generated \\ \hline
$X.\textit{expiration}$ & the real-time when $X$ expires (\emph{i.e.}, $\invt{X.\textit{owner}}{E_{X.\textit{owner}}}$)
\end{tabular}
\vspace{1em}

\noindent
Some trivial observation about such an $X$ are:

\begin{equation}
X.\textit{start} \le X.\textit{completion} < X.\textit{expiration}
\end{equation}
\begin{equation}\label{eq:leadership}
\forall t: (X.\textit{completion} \le t < X.\textit{expiration}) \Rightarrow
					\textit{isLeader}_{X.\textit{owner}}(t)
\end{equation}
\begin{equation}\label{eq:interval}
\forall q, T: (q, T) \in X.\textit{QT} \Rightarrow X.\textit{start} \le \invt{q}{T} \le X.\textit{completion}
\end{equation}

\noindent
We order instantiations by their completion time, that is,
$X < X' \Leftrightarrow X.\textit{completion} < X'.\textit{completion}$.

\begin{lemma}
$\forall X, X': X < X' \Rightarrow X.\textit{QT} < X'.\textit{QT}$.
\end{lemma}
\begin{proof}
By contradiction, assume there can exist an $X$ and $X'$ such that
$X < X'$ (and thus $X.\textit{completion} < X'.\textit{completion}$)
and $\lnot(X.\textit{QT} < X'.\textit{QT})$.
Because quorums overlap, there must exists a $q$, $T$, $T'$ such that
$(q, T) \in X.\textit{QT}$ and $(q, T') \in X' .\textit{QT}$.
By assumption, $T \ge T'$ (for otherwise $X.\textit{QT} < X'.\textit{QT}$).
We consider two cases.
\begin{itemize}
\item $X.\textit{owner} = X'.\textit{owner}\,$: Then it must be the
case that $X.\textit{completion} < X'.\textit{start}$ (or $X$ would
have been aborted and could not have completed).
From \eqnref{eq:interval} it must be that $T < T'$, contradicting
the assumption that $T \ge T'$.

\item $X.\textit{owner} \ne X'.\textit{owner}\,$:
From time $\invt{q}{T}$ until $X.\textit{completion}$ (and beyond),
$q$ has granted a lease to $X.\textit{owner}$, and similarly,
from $\invt{q}{T'}$ to $X'.\textit{completion}$,
$q$ has granted a lease to $X'.\textit{owner}$.
Because $X.\textit{completion} < X'.\textit{completion}$ and
$\invt{q}{T} \ge \invt{q}{T'}$, it must be the case that at time
$X.\textit{completion}$, process $q$ has granted a lease both to
$X.\textit{owner}$ and $X'.\textit{owner}$.  But a process cannot
grant leases to two different processes at the same time. \hfill\qed

\end{itemize}
\end{proof}

Note, as a corollary, that quorum timestamps are well-ordered,
consistent with the ordering on instantiations of \texttt{OQwBD}.

\section{Leader Stability}\label{sec:stability}

When there is a leader, Leader Stability implies that the leader persists
in that role in the absence of failures and while
messages are delivered and processed in a timely fashion.
Suppose that message round-trip time is bounded by a known constant $d$.
In that case, if leader $p$ starts \texttt{OQwBD} before $c_p(E_p) - d$,
then it is able to extend its leadership before it expires.
So $p$ should use $\delta > 2d$ in order that in the
next period of leadership it is able to do so again.
Choices of $\delta$ are discussed in Section~\ref{sec:delta}.

\begin{comment}
If we assume there is a time \texttt{GST}
after which there is such a bound $d$,
and after which there are no new failures,
then it is easy to show that Leader Stability holds.
In practice, a protocol will need to estimate the round-trip time
from prior instantiations of \texttt{OQwBD}.

This reasoning tacitly assumed that no other process is trying to become
leader at the same time.
If more than one process tries to become leader, they can interfere with
one another in such a way that none are able to get a quorum.
In that case, they would have to wait to let grants expire.
In order to try to avoid this situation, a process $q$ should not initiate
\texttt{OQwBD} if $A_q \ne q \wedge C_q < F_q$ (\emph{i.e.}, if it is
granting a lease to some other process $A_q$ at the time).

In \texttt{OQwBD}, processes only respond to the initiator if they
grant the lease.  If a process has granted another lease to another process,
it could respond with an error message, and specify how much longer that
previous request is valid.  This, in turn, would help the initiator if
there is any hope of still obtaining a quorum, and if not, how much longer
to wait before retrying.
\end{comment}

\section{Eventual Election}\label{sec:eventual}

If there is no leader (\emph{i.e.}, $\forall p: C_p \ge E_p$) then
multiple processes could try to become leader.
There is no guarantee that any will succeed, however.
But if we could somehow ensure that only one process $p$ executes
\texttt{OQwBD}, and the process waits long enough to do so (so that all
$F_*$'s have expired), then \texttt{OQwBD} is guaranteed to succeed
eventually (after \texttt{GST}).
This would seem to create a circularity, as
as choosing $p$ requires solving leader election.
The way out is to use a weak version of leader election (which may select
multiple \emph{weak leaders})
in order to make successful completion of \texttt{OQwBD} likely.
The more likely it is that weak leader election selects only a single
weak leader, the more likely an instantiation of \texttt{OQwBD} terminates
successfully.
In addition, for Eventual Election to hold, after \texttt{GST} the weak leader
election protocol is required to produce a single weak leader.

Here is such a weak leader election algorithm:  Assume
processes in $P$ are ordered, that is, $p < q < ...$,
and elect the smallest process in $P$ that has not failed.
To this end, processes are organized into a virtual ring in order.
The scheme uses a failure detection algorithm such as simple pinging
or the more sophisticated $\phi$-accrual
failure detector~\cite{HDYK04}
that gives a better approximation of the failure status of processes.
Each process pairs with the closest predecessor and closest
successor on the ring that it considers correct by the failure detector,
and monitors it.
If a process $q$ believes it is the lowest correct process (because the
identifier of its predecessor is larger than its own),
then it considers itself a weak leader.
Note that under the properties of \texttt{GST}, failure detection becomes
accurate and the algorithm will produce a single leader.

A weak leader initiates \texttt{OQwBD} to try to become
a leader if
$A_q \ne q \wedge C_q < F_q$ (\emph{i.e.}, it is not currently granting
a lease to another process).
It does so periodically in order to deal with possible collisions and
message loss.

\begin{comment}
We will require of the weak leader election protocol that after \texttt{GST},
eventually, the protocol identifies a single leader.
Note that if there is a bound $d$ on message round-trip and processing
times, then even the simple pinging protocol will eventually produce a
single leader as a process that does not respond within $d$ to a
ping is guaranteed to be faulty.
\end{comment}

The Eventual Election property
states that under the conditions that hold after \texttt{GST},
a leader will eventually be chosen by the Nerio protocol if there is none yet
(and, because of Leader Stability, it will remain leader henceforth).
To see why Eventual Election holds for the presented protocols, note
that eventually only one process will attempt to become leader because
of the properties of weak leader election.
After all conflicting grants have expired, and because
round-trip latencies are bounded by $d$, eventually this process
will complete \texttt{OQwBD}.

\section{Discussion}\label{sec:discussion}

\subsection{Choice of $\delta$}\label{sec:delta}

A process that initiates \texttt{OQwBD} chooses some value for $\delta$.
No matter what value of $\delta$ is chosen,
Leader Uniqueness will hold, but choosing $\delta$ too small could
adversely affect Leader Stability and Eventual Election.
Therefore $\delta$ should be chosen large enough so a leader can
extend the period of its leadership without interruption,
but short enough so that recovery can be swift after the leader fails.

Suppose $d$ is an estimate for the round-trip time, and represents
that, say, in 99.9\% of round-trips the round-trip latency is less than $d$.
A leader might initiate \texttt{OQwBD} to extend its lease
before $E_p - (1 + \rho) \cdot d$ measured on its local clock.
Clearly, $\delta$ should be chosen larger than $d$ plus the time that
remains on the lease, which can be conservatively estimated by $p$ as
$(1 + \rho) \cdot (E_p - C_p)$.

In practice, $d$ is likely no larger than a few milliseconds on today's
hardware, assuming processing of Nerio messages receive a high priority,
and $\rho$ is likely no more than 10 microseconds.
But choosing $\delta$ as small as possible would likely result in too many
round-trips per second.  If we want space out instantiations of \texttt{OQwBD}
by at least $i$ time units, then we should choose
$\delta = d + \max(i, (1 + \rho) \cdot (E_p - C_p))$.

\subsection{Interference}

If more than one process concurrently tries to become leader,
then none may be able to enlist a quorum.
They each would then have to wait to let conflicting grants expire
before attempting to rerun \texttt{OQwBD}.
\begin{comment}
A process $q$ should not initiate \texttt{OQwBD}
if $A_q \ne q \wedge C_q < F_q$ (\emph{i.e.}, if it is
granting a lease to some other process $A_q$ at the time), but this
does not eliminate interference altogether.
\end{comment}

In both proposed Nerio protocols,
processes respond to the initiator only if they
grant the lease request.
But there is something to gain if the protocols are modified so that
if a process has granted a lease to another process,
then instead of just ignoring the grant request,
it responds with an error message.
The error message helps an initiator to determine if
there is hope of obtaining a quorum.

The protocols can be extended with
revocation requests to further avoid interference.
This further extension requires that grant and revocation requests from the same
source are delivered in FIFO order.
When a process $q$ receives a revocation request from a process $p$,
and if $A_q = p \wedge C_q < F_q$, then $q$ sets $F_q$ to $C_q$,
thereby releasing its grant to $p$.
(The FIFO order ensures that delayed revocation request do not
inadvertently revoke outstanding grants.)
Obviously, a process $p$ that sends a revocation request must first have
aborted the protocol, thus even if it ends up
collecting positive responses from a quorum, $E_p$ should not be advanced.

\subsection{Network Partitioning}

Nerio class protocols work work even if the network partitions if
there is a partition that contains a quorum of correct processes.
And if there is no such partition or if functionality is desired in minority
partitions (\emph{i.e.}, partitions that do not hold a quorum of processes),
then the weak leader election algorithm might
be used to assign a temporary, non-authoritative, leader in each partition
that can provide partial functionality.  % Obvious caveats apply.

\subsection{Finding the Leader}

What if an external client, seeking that an edict be issued,
sends a request to a process $p$ in $P$ but $p$
is not currently leader?  If process $p$ has an unexpired
grant for another process $q$, then process $p$ can respond by giving $q$
as a forwarding address.  If not, process $p$ may attempt to become leader.
Failing that, $p$ may buffer the request until a leader emerges,
or return an error response.

\subsection{Leader Verification}

A process $q$ may want to check whether some other process $p$ is leader.
The following protocol, based on bounded drift, will accomplish this:

\begin{enumerate}
\item set $\textit{Start}_q := C_q$ (save starting time);
\item send $\langle \texttt{verifyLeadership}, q, \textit{Start}_q \rangle$
								to $p$.
\end{enumerate}

\noindent Upon receipt of a \texttt{verifyLeadership} message,
a process $p$ does the following:

\begin{enumerate}
\setcounter{enumi}{2}
\item calculate $\delta := (E_p - C_p) / (\rho + 1)$;
\item send $\langle \texttt{remainder}, p, \delta, \textit{Start}_q \rangle$ to $q$.
\end{enumerate}

Here $\delta$ equals the minimal amount of real time that is left of $p$'s
leadership.  Note that if $p$ is no longer leader, $\delta$ will be negative.
If $q$ receives the response, it calculates
$T = \textit{Start}_q + \delta \cdot (\rho - 1)$, and as long
as $T < C_q$ holds, $p$ is guaranteed to be leader (and possibly a bit longer
than that depending on rate differences between $C_p$ and $C_q$).

\subsection{Changing Membership}

Nerio protocols can be adapted to handle the case where $P$ changes over time.
We introduce \emph{epochs}, numbered consecutively starting at 0.
Each epoch $e$ is associated with a set of processes $P_e$ and
quorum system ${\cal Q}_e$ defined on $P_e$.
For simplicity, assume different epochs have non-overlapping sets of
processes:

\begin{equation}
\forall e, e': e \ne e' \Rightarrow P_e \cap P_{e'} = \emptyset.
\end{equation}

\noindent
(In practice, a process that is a member of more than one epoch
should maintain different copies of its state variables for each epoch.)

Each epoch is defined to be \texttt{PENDING}, \texttt{RUNNING},
or \texttt{TERMINATED}.
Each epoch starts in the \texttt{PENDING} state, except for epoch~0
which starts in the \texttt{RUNNING} state.
At any point in time, at most one epoch is in the \texttt{RUNNING} state,
and all prior epochs are \texttt{TERMINATED}.

If epoch $e$ is \texttt{RUNNING}, it can be terminated by getting
each process $q$ in some quorum of ${\cal Q}_e$ to set
$A_q = \bot \wedge F_q = \infty$.
(A process can only do so if there is no current outstanding grant.)
We say that process $q$ is \emph{wedged} if $A_q = \bot \wedge F_q = \infty$
holds.
Once a quorum of processes are wedged,
no process can become leader in that epoch.
At this same time, epoch $e+1$ automatically becomes \texttt{RUNNING}.
That is, an epoch is defined to be \texttt{RUNNING} iff all prior epochs are
\texttt{TERMINATED} and no quorum of processes are all wedged in that epoch.

A process $p$ in epoch $e+1$ ignores grant requests, and does not send any,
until it has learned that epoch $e$ is \texttt{TERMINATED}.
Process $p$ can learn that $e$ is \texttt{TERMINATED}
by querying processes in a quorum of ${\cal Q}_e$ and detecting that
these are wedged,
or by receiving a grant request from a process in $P_{e + 1}$.

Note that $\textit{isLeader}_p(t)$ will hold only if
epoch $e$ is \texttt{RUNNING} at time $t$ and $p \in P_e$ holds.
Because at most one epoch is \texttt{RUNNING}, Leader Uniqueness continues
to hold, even given multiple epochs.

Leader Stability no longer makes sense because epoch memberships are
non-overlapping.  However, an epoch $e+1$ that wants to start running could
have a particular process $p \in P_{e+1}$
be in charge of wedging the processes in $P_e$, by sending a
$\langle \texttt{grantRequest}, \bot, \textit{Start}_p, \infty \rangle$ message
to these processes,
and upon obtaining \texttt{ok} responses from a quorum of those processes,
send a regular grant request to the processes of epoch $P_{e+1}$.
Thus the new epoch has significant control over which process it wants to
be leader initially.

Once a process receives a grant request with $\delta = \infty$, but has
an outstanding (normal) grant request, it could buffer the special
grant request and grant it upon expiry of the current lease.  In that
case, if at most a quorum of processes in a \texttt{RUNNING} epoch
are faulty, eventually the epoch will become \texttt{TERMINATED} and
the processes of the next epoch will be able to learn so.
Therefore Eventual Election also continues to hold.

The reconfiguration protocol should be invoked when processes are
suspected of having crashed, or eventually there may no longer be
a quorum available to elect a leader.
The reconfiguration protocol can only make progress if a quorum in
$Q_e$ is correct and can be wedged, and thus if too many processes
crash, it is no longer possible to reconfigure.
Under manual intervention, an administrator could explicitly mark
certain processes as having failed.
The quorum system could then be adjusted with smaller quorums in
order to make progress.

Note that edict timestamps can be extended with epochs in order to
make sure the Edict Ordering continues to hold.

\section{Related Work}\label{sec:related}

Leader election is used in practical systems.
For example, the IEEE 1394 ``Firewire'' serial bus standard,
for the purpose of coordination among devices, includes such a protocol that
creates a spanning tree of devices with a unique root acting as leader.
Early work on leader election focused on efficiently finding extremas
(the node with the minimum or maximum identifier) in a connected network
topology of unknown size.
The problem was apparently first formulated and solved in 1977
by Gerard LeLann~\cite{Lel77}.
Many papers on this subject have appeared since.

In 1982, Hector Garcia-Molina defined the problem of leader election
in a distributed system that admits failures~\cite{GM82}, and presented
protocols.  That paper includes
separate definitions for synchronous and asynchronous systems.  For
a synchronous system, Garcia-Molina's definition of leader election
requires that there be at most
one leader at at time, and in the absence of failures a leader is elected 
within a fixed time limit.  For an asynchronous system, the definition
applies only to those nodes that experience synchronous communication---the
other nodes may end up with different leaders.

Consensus protocols~\cite{BM93} can be used to solve leader
election in both synchronous and asynchronous systems.
Each participant proposes itself as leader, and the consensus protocol
subsequently decides on one of the proposals.  Dividing time into time
slots, an instantiation of consensus could be used for each time slot.
Doing so would lead to unnecessarily high overhead, and many
consensus protocols rely on leader election themselves, creating a
circularity.

Fueled by leader-based consensus protocols, many papers discuss
leader election in partially asynchronous systems.
In this formulation, a protocol may output multiple leaders, but there
must exist a time after which the protocol output exactly one leader.
In asynchronous environments these protocols are probabilistic,
producing a single leader in case the environment is reasonably timely,
but that may produce multiple leaders in case the environment is not.
We call this \emph{weak leader election}, but it is also referred to
as \emph{local leader election}.

Weak leader election in asynchronous environments is
closely related to the failure detection problem, whereby a leader
is the node with the lowest (or highest) identifier that is not
suspected of having failed.
Fetzer and Cristian~\cite{FZ99}
study the problem of weak leader election in partitionable networks,
and use a technique based on leases~\cite{GC89}
(to define partition boundaries).
Stable (but weak) leader election was considered in~\cite{ADFT01}.
A performance comparison of three recent stable leader election algorithms
appears in~\cite{ST08}.  This paper also consider dynamic membership.

The problem of strong leader election in a partially synchronous
environment was discussed by Fetzer and S\"{u}{\ss}kraut in~\cite{FS06}.
Their protocol uses leases and quorums.
The Nerio protocols described in this paper
generalize this idea by defining an invariant (Equation~\eqnref{eq:invariant})
that all such protocols must satisfy, and can be used to transform any weak
leader election protocol into one that is both strong and stable,
and support dynamic membership.

\bibliographystyle{plain}
\bibliography{all}

\end{document}